\documentclass[journal,twoside,web]{ieeecolor}
\usepackage{lcsys}
\usepackage{cite}
\usepackage{caption}
\usepackage{amsmath,amssymb,amsfonts}
\usepackage{algorithmic}
\usepackage{graphicx}
\usepackage{textcomp}
\usepackage{bm}
\usepackage{soul}
\usepackage{nicefrac}

\usepackage{babel,blindtext}
\usepackage{mathtools}  
\usepackage{commath}
\usepackage{array}
\usepackage{derivative}
\usepackage{enumerate}
\usepackage[colorlinks,allcolors=black]{hyperref} 
\usepackage[capitalize,nameinlink]{cleveref}

\newtheorem{theorem}{Theorem}
\newtheorem{lemma}{Lemma}

\newtheorem{corollary}{Corollary}[lemma]

\newtheorem{remark}{Remark}
\newtheorem{assumption}{Assumption}

\usepackage[usenames,dvipsnames]{xcolor}
\usepackage{tikz}
\usetikzlibrary{positioning}
\usetikzlibrary{decorations.pathreplacing,calligraphy}

\usepackage{stackengine}
\newcommand\barbelow[1]{\stackunder[1.2pt]{$#1$}{\rule{.8ex}{.075ex}}}

\def\BibTeX{{\rm B\kern-.05em{\sc i\kern-.025em b}\kern-.08em
    T\kern-.1667em\lower.7ex\hbox{E}\kern-.125emX}}
\usepackage[colorlinks,allcolors=black]{hyperref} 
\usepackage[capitalize,nameinlink]{cleveref}

\definecolor{blue}{rgb}{0.13, 0.67, 0.8}
\definecolor{purple}{rgb}{0.41, 0.16, 0.38}
\definecolor{black}{rgb}{0, 0, 0}
\DeclareRobustCommand\leb  {\tikz[baseline=-0.6ex]\draw[thick, color=blue] (0,0)--(0.2,0);}
\DeclareRobustCommand\lep  {\tikz[baseline=-0.6ex]\draw[thick, color=purple] (0,0)--(0.2,0);}
\DeclareRobustCommand\ledas{\tikz[baseline=-0.6ex]\draw[thick,dashed] (0,0)--(0.3,0);}
\DeclareRobustCommand\leblack{\tikz[baseline=-0.6ex]\draw[thick, color=black] (0,0)--(0.3,0);}

\usepackage{type1cm}
\usepackage{eso-pic}
\usepackage{color}

\makeatletter
\AddToShipoutPicture{%
            \setlength{\@tempdimb}{.5\paperwidth}%
            \setlength{\@tempdimc}{.5\paperheight}%
            \setlength{\unitlength}{1pt}%
            \put(\strip@pt\@tempdimb,\strip@pt\@tempdimc){%
        \makebox(0,0){\rotatebox{45}{\textcolor[gray]{0.92}%
        {\fontsize{6cm}{6cm}\selectfont{Preprint}}}}%
            }%
}

\begin{document}
\title{Approximation-free Prescribed Performance Control with Prescribed Input Constraints}
\author{Pankaj K Mishra and Pushpak Jagtap, \IEEEmembership{Member, IEEE}
\thanks{Article is submitted to IEEE Control Systems Letters.}
\thanks{This work was partly supported by the Google Research Grant, the SERB Start-up Research Grant, and the CSR Grant by Nokia Corporation.}
\thanks{The authors are with the Robert Bosch Centre for Cyber Physical Systems,
Indian Institute of Science, Bangalore 560012, India (e-mail:
\{pankajmishra,pushpak\}@iisc.ac.in)}}
\maketitle
\begin{abstract}
This paper considers the tracking control problem for an unknown nonlinear system with time-varying bounded disturbance subjected to a prescribed performance and input constraints. When performance and input constraints are specified simultaneously for such a problem, a trade-off is inevitable. Consequently, a feasibility condition for prescribing performance and input constraints is devised to address such difficulties of arbitrary prescription. In addition, an approximation-free controller with low complexity is proposed, which ensures that the constraints are never violated, provided that the feasibility condition holds. Finally, simulation results corroborate the effectiveness of the proposed controller.
\end{abstract}

\section{INTRODUCTION}

Different methods have been developed by academic and professional researchers for designing controllers for nonlinear systems. Despite these efforts, designing controllers for systems subjected to constraints and unknown time-varying disturbances remains challenging. Various constraints arise in most practical systems, including performance constraints, saturation, physical stoppages, and safety requirements. Therefore, constraints cannot be avoided while designing controllers for practical systems. For controller design, constraints can typically be prescribed in two forms: prescribed performance constraints (PPC) on some variable (such as tracking error) and prescribed input constraints (PIC). 

A wide variety of methods have been developed to address PPC, including: reference governors \cite{rg}, model predictive control \cite{mpc}, funnel control \cite{fc}, barrier Lyapunov functions (BLF) \cite{Mishra2021f}, prescribed performance control \cite{Bechlioulisl2013}, control barrier functions \cite{cbf}, and extremum seeking control \cite{es}.  As far as the literature is concerned, BLF has been extensively used in dealing with constraints. That's because its design methodology allows it to incorporate many Lyapunov-based nonlinear control techniques. However, for uncertain, unknown systems, the design lacks simplicity. The approach described in \cite{Bechlioulisl2013} is well-suited to a diverse set of situations \cite{Zhang2019a} since it is low-complexity, approximation-free, and robust. Much research and development have gone into the controller design for nonlinear systems that undergo input saturation. One can refer to the results in  \cite{Wen2011,Chen2015,Xiao2012}. It follows that controller design for nonlinear systems subjected to either PPC or PIC is a well-established field of study.

Improving performance with limited resources is always difficult. Same with PPC and PIC \cite{Yong2020}, PPC aims for lower steady-state error, safe transient response, and fast convergence of tracking error. In contrast, PIC focuses on actuator safety or control effort minimization. Thus, very few results are available addressing PPC and input saturation, notably \cite{Hopfe2010,Hopfe2010a,berger2022input,Kanakis2020}. In \cite{Hopfe2010}, works are done for linear systems,  nonlinear systems in \cite{Hopfe2010a,Kanakis2020}.  Also, in  \cite{Hopfe2010,Hopfe2010a, berger2022input} authors relax the PPC  whenever the input saturation is active, and in \cite{Kanakis2020} assumptions are made on the existence of a feasible set of control input for a given initial conditions and actuator saturation limit. Moreover, given any desired trajectory for an uncertain nonlinear system with unknown bounded disturbances and arbitrary PIC, it is certainly impractical to guarantee that the desired trajectory is trackable. For example, a large external disturbance or a desired trajectory with a large upper bound will inevitably necessitate the same level of opposing control command, which may extend beyond the PIC \cite{Asl2019}. Thus, before prescribing input constraints, one must look for the feasible condition for  PIC. Further,  many practical systems always operate in some specified regions where they are controllable under PIC \cite{Yong2020}. In the presence of PIC, one cannot globally stabilize the unstable system. There is always a feasible set of initial conditions for PIC. Also, global results are not attained in many PPC studies  \cite{Theodorakopoulos2015,Bechlioulis2017,Zhang2017} on tracking error. The prescribed performance function choice depends on the constrained variable's initial state. In \cite{Cao2022g}, global results were achieved by transiently relaxing the PPC. However, as discussed, it makes no sense to pursue global results when there is a PIC. In addition, arbitrary PPC makes no sense because there may be an initial condition of error variable within the initial bounds of PPC that does not belong to the set of initial conditions that are feasible for PIC. Therefore, we must seek a viable PPC for a PIC.

  Motivated by the above discussions and aforementioned works, a controller has been developed in this paper with the following listed contributions: 
\begin{enumerate}
    \item An approximation-free low complexity controller has been proposed for the nonlinear system with PPC and PIC. {Below is a representation of the controller structure:} $$\frac{2\bar \upsilon}{\pi}\arctan\left(\frac{\pi}{2\bar \upsilon}\tan\left(\frac{\pi e}{2\psi}\right)\right),$$ where $\bar \upsilon$ and  $\psi$ represent PIC and PPC respectively, and $e$ is tracking error.
 \item The above novel control structure consists of  both PPC and PIC in its design. This simplifies the task of deriving a feasibility condition for PPC and PIC, avoiding the need to relax PPC when input approaches its constraints.
\end{enumerate}
The remainder of the paper is structured as follows. In Section II, preliminaries and problem formulation are presented. It also contains key assumptions for the prescription of input constraints. Section III presents the design of the controller. Section IV presents the few lemmas used for the stability analysis in Section V.  Section VI presents  the simulation results and discussion. Finally, Section VII concludes the paper.

\section{Preliminaries and Problem Formulation}
\textbf{Notations:} We denote the set of real, positive real, nonnegative real, and positive integer numbers by $\mathbb{R}$, $\mathbb{R}^+$, $\mathbb{R}_0^+ $, and $\mathbb{N}$, respectively. $\mathbb{N}_n$: $\{1,\ldots , n\}$, $n$ is positive integer. $\mathcal{L}^\infty$ represents the set of all essentially bounded measurable functions. For $x (t)\in \mathbb{R}$, $x\uparrow a$: $x$ approaches a real value $a$ from the left side,  $x\downarrow a$:  $x$ approaches a real value $a$ from the right side, and $x^{(n)}$ represent $n$th time derivative of signal $x$.  $\binom{m}{k}=\frac{m(m-1)\cdots(m-k+1)}{k!}$ denote the binomial coefficients. $L\{\cdot\}$ denotes the Laplace transform, and $s$ is the Laplace variable.  

Consider a class of strict-feedback nonlinear system 
\begin{equation}\label{sys1}
\begin{split}
\dot \xi_i&=\xi_{i+1}, ~\forall i \in \mathbb{N}_{n-1},\\
\dot \xi_n&=f\left({\bm \xi}\right)+g\left({\bm \xi}\right)\upsilon+d,\\
y&=\xi_1,
\end{split}
\end{equation} 
where  $ \bm{\xi}(t)=[\xi_1(t), \ldots, \xi_n(t)]^T \in \mathbb{R}^n$ is the state vector,  $f:\mathbb{R}^{n}\rightarrow\mathbb{R}$ is the unknown smooth nonlinear function, $g:\mathbb{R}^{n}\rightarrow\mathbb{R}$ is the unknown control coefficient, $d(t) \in \mathbb{R}$ is the unknown piecewise continuous bounded disturbance, $\upsilon (t)  \in \mathbb{U}\subseteq \mathbb{R}$ and $y(t) \in \mathbb{R}$ are the input and output of the system, respectively. 

The control problem is to design a control law $\upsilon$ such that $(i)$ the output $\xi_1(t)$ track the desired output $\xi_\mathsf{d}(t)\in \mathbb{R},~ \forall t \in \mathbb{R}^{+}_{0}$, $(ii)$ output tracking error defined as  $\Tilde\xi\coloneqq\xi_1-\xi_\mathsf{d},$ follow its prescribed performance constraints $\psi(t)\in \mathbb{R}$, defined as   $\psi(t):=\psi_0e^{-\mu t}+\psi_\infty,$ such that $|\Tilde\xi|<\psi(t),~ \forall t \in \mathbb{R}^{+}_{0},$ where $\psi_0$ is a positive constant, and $\psi_\infty$ and $ \mu$ are positive and nonnegative constants,  represent the bounds on the steady-state error and the decay rate of the tracking error, respectively, and $(iii)$ all the closed-loop signals are bounded.

In addition, one of our problems will be to seek the feasibility condition for the PIC and PPC. Obtaining such a feasibility condition will necessitate specific knowledge of the system's dynamics, disturbances, and tracking performance parameters in terms of their upper bounds on the signals. A few assumptions are required for this are listed below.
\smallskip

\begin{assumption}\label{af}
\cite{7150402,4717255, 6075525, 7289399} The unknown map $f$ satisfies the Lipschitz continuity condition, that is, for all $\bm x,\bm x'\in \mathbb{R}^n$, there exists a constant $k_l\in\mathbb{R}^+$ such that the following holds 
\begin{align*}
 |f(\bm x)-f(\bm x')|\le k_l||\bm x-\bm x'||_{p^*},
\end{align*}
where $k_l$ is a known Lipschitz constant and $||\cdot||_{p^*}$ known as the $p^*$ norm in the $\mathbb{R}^n$.
\end{assumption}
Note that one can use the Lipschitz constant inference approaches proposed in \cite{wood1996estimation,Bubeck2011,malherbe2017global} to estimate the Lipschitz constant of unknown dynamics from a finite number of data collected from the system.
\smallskip
\begin{assumption}\label{ag}
There exist a known constant $\barbelow g>0$ and a  constant $\bar g\ge\barbelow g$, such that $\barbelow g\le g(x)\le\bar g$ for all $x \in\mathbb{R}^n$.
\end{assumption}
\smallskip
\begin{assumption}\label{ad}
There exists known constant $\bar d\ge 0$ such that disturbances $|d(t)|\le \bar d$ for all $t\in \mathbb R_0^+$.
\end{assumption}
\smallskip
\begin{assumption}\label{ade}
For a given desired trajectory $\xi_\mathsf{d}$, there exists a constant $\bar\xi_\mathsf{d}>0$, such that  $||\bm{\xi_\mathsf{d}}(t)||_{\infty}<\bar\xi_\mathsf{d},$ for all $t\in\mathbb{R}^+_0$ for $\bm{\xi_\mathsf{d}}=[\xi_\mathsf{d}, \xi_\mathsf{d}^{(1)}, \hdots,\xi_\mathsf{d}^{(n-1)}]^T.$ 
\end{assumption}

\section{Controller Design}
This section proposes a robust approximation-free controller for \eqref{sys1}. To begin the controller design, we define a filtered tracking error,
\begin{align}
    r\coloneqq\lambda_1\tilde\xi+\lambda_2\dot{\tilde\xi}+\cdots+\lambda_{n-1}{\tilde\xi}^{(n-2)}+{\tilde\xi}^{(n-1)}, \label {r}
\end{align}
where  $\lambda_i, \forall i \in \mathbb{N}_{n-1}$ is a strictly positive constant and following the definition of output tracking error mentioned in the problem statement,
\begin{align}\label{tilde}
    {\tilde\xi}^{(i-1)}=\xi_i-\xi_\mathsf{d}^{(i-1)},~\forall i \in \mathbb{N}_{n}.
\end{align}
Taking the time derivative of \eqref{r} and using \eqref{tilde}, one has
\begin{align}
    \dot r=\lambda_1{\tilde\xi}^{(1)}+\lambda_2{\tilde\xi}^{(2)}+\cdots+\lambda_{n-1}{\tilde\xi}^{(n-1)}+\dot\xi_n-\xi_\mathsf{d}^{(n)}. \label {r1}
\end{align}
Using \eqref{sys1} and \eqref{r1}, closed-loop dynamics can be written as
\begin{align}
    \dot r= \phi+f\left({\bm \xi}\right)+g\left({\bm \xi}\right)\upsilon+d-\xi_\mathsf{d}^{(n)}, \label{r2}
\end{align}
where $\phi=\sum_{i=1}^{n-1}\lambda_i{\tilde\xi}^{(i)}$.

Consider a non-increasing smooth function $\psi_r:\mathbb R_0^+\rightarrow\mathbb R^+$  as a  virtual performance constraint (VPC)  over $r$, defined as 
\begin{align}\label{constr}
    \psi_r(t):=\psi_{r0}{e}^{-\mu_r t} +\psi_{r\infty}, \forall t\in\mathbb{R}_0^+,
\end{align}
where $\psi_{r\infty},\psi_{r0}$ and $\mu_r$ have similar attributes  as of $\psi_{\infty},\psi_{0}$ and $\mu$ for PPC.
Note that, in \eqref{constr},  $\psi_r$ and $\dot \psi_r$ are bounded for all $t\in\mathbb{R}_0^+$ and the bounds are given as
\begin{align}
    \psi_{r\infty}&\le{\psi_r}\le \psi_{r0}+\psi_{r\infty}, \text{~and} \label{psib}\\
    -\mu_r\psi_{r0}&\le\dot \psi_r \le 0. \label{psidb}
\end{align}
 The control input is designed as 
\begin{align}\label{u}
    \upsilon=-\frac{2\bar \upsilon}{\pi}\arctan\left(\cfrac{\pi}{2\bar\upsilon}\tan\left(\frac{\pi r}{2\psi_r}\right)\right),
\end{align}
where $\bar\upsilon$ is PIC, $r$ and $\psi_r$ are as mentioned in \eqref{r}, respectively.
In \eqref{u}, $r$ is designed using \eqref{r} and \eqref{tilde}, with 
\begin{align}
    \lambda_i&=\binom{n-1}{n-i}a^{n-i}, ~a>\mu,~\forall i \in \mathbb{N}_{n-1},\label{lambda}
\end{align}
and  $\psi_r,$ i.e., VPC is chosen based on  the PPC defined in the problem statement, as follows
\begin{align}
      \mu_r&=\mu,\label{mu}\\
      \psi_{r0}&=(a-\mu_r)^{n-1}\psi_0,\label{psir0}\\
      \psi_{r\infty}&=a^{n-1}\psi_\infty.\label{psirinf}
\end{align}


\section{Preliminaries the Stability Analysis}
In this section,  first, a few results will be established, which will motivate the idea behind the selection of parameters of VPC $(\psi_r)$ in \eqref{mu}-\eqref{psirinf} based on PPC $(\psi)$. Further, a few lemmas will be presented, which will be later used in stability analysis. The lemmas are as follows.
\smallskip
\begin{lemma}\label{lem1}
    Consider the  signals $X(t) \in \mathbb{R}$ and $Z(t) \in \mathbb{R}$, such that  $|X(t)|<X_0e^{-\mu_xt}+X_\infty,$ where  have similar attributes  as of $\psi_{\infty},\psi_{0}$ and $\mu$ for PPC. If   $z=\frac{x}{(s+a)^p},$ where $z=L\{Z(t)\}$, $x=L\{X(t)\}$,  and $a>\mu_x$ and $p\in \mathbb{Z}^+$, then $|Z(t)|<Z_0e^{-\mu_xt}+Z_\infty$, with $Z_0=\frac{X_0}{(a-\mu_x)^p}$ and $Z_\infty=\frac{X_\infty}{a_{}^p}.$
\end{lemma}
\begin{proof}
Here,  $z=\frac{x}{(s+a)^p},$ can be represented as a signal passing through a series of low pass filters as shown in the figure below:
\smallskip
\begin{center}
\begin{tikzpicture}
\node[draw,
    minimum width=1cm,
    minimum height=1cm](f1){$\frac{1}{s+a}$};
    
\node[draw,
minimum width=1cm,
minimum height=1cm,
right=0.5cm of f1](f2){$\frac{1}{s+a}$};

\node[draw,
minimum width=1cm,
minimum height=1cm,
right=2cm of f2](f3){$\frac{1}{s+a}$};

\draw [stealth-](f1.west)--++(-0.5,0)     
    node[midway,above]{$x$};
\draw [-stealth](f1.east)  -- (f2.west)  
    node[midway,above]{$z_1$};   
\draw [-stealth](f2.east)  -- ++(1,0)  
    node[midway,above](a3){}; 
\draw [stealth-](f3.west) -- ++(-0.5,0)  
 node[midway,above](a4){};
 \draw [-stealth](f3.east) -- ++(0.5,0)  
 node[midway,above]{$z$};
\path (f2) -- node[auto=false]{\ldots} ++(3,0);
\draw [decorate,
    decoration = {calligraphic brace,mirror}] (-0.5,-0.8) --  (5,-0.8) node[midway,below]{$p$ blocks};
\end{tikzpicture}
\end{center}
Let $z_1$ be the output of first filter, then $Z_1(t)= L^{-1}(z_1)$, can be written as $Z_1(t)=\int_0^t e^{-a(t-\tau)}X(\tau)d\tau.$
Since $|X(t)|<X_0e^{-\mu_xt}+X_\infty,$ thus we have $|Z_1(t)|<\int_0^t e^{-a(t-\tau)}(X_0e^{-\mu_x\tau}+X_\infty)d\tau.$
Simplifying it, we have $|Z_1(t)|<\frac{X_0}{(a-\mu_x)}(e^{-\mu_xt}-e^{-at})+ \frac{X_\infty}{a}(1-e^{-at}).$
Further it can be written as $ |Z_1(t)|<\frac{X_0}{(a-\mu_x)}e^{-\mu_xt}+\frac{X_\infty}{a}.$

Recursively following the above steps $(p-1)$ times, it can be easily found that $|Z(t)|<\frac{X_0}{(a-\mu_x)^p}e^{-\mu_xt}+\frac{X_\infty}{a_{}^p}.$
\end{proof}
\smallskip
\begin{corollary}\label{cor1}
If in Lemma \ref{lem1}, $z=\frac{s^q}{(s+a)^q}x, q\in \mathbb{Z}^+$; then $|Z(t)|<Z_0e^{-\mu_xt}+Z_\infty$, with $Z_0={X_0}\left(\frac{2a-\mu_x}{a-\mu_x}\right)^q$ and $Z_\infty=2^qX_\infty.$
\end{corollary}
\begin{proof}
Similar to Lemma \ref{lem1}, $z=\frac{s}{(s+a)^q}x,$ can be represented as a signal passing through a series of filters as shown in the figure below.
\smallskip
\begin{center}
\begin{tikzpicture}

\node[draw,
    minimum width=1cm,
    minimum height=1cm](f1){$\frac{s}{s+a}$};
    
\node[draw,
minimum width=1cm,
minimum height=1cm,
right=0.5cm of f1](f2){$\frac{s}{s+a}$};

\node[draw,
minimum width=1cm,
minimum height=1cm,
right=2cm of f2](f3){$\frac{s}{s+a}$};

\draw [stealth-](f1.west)--++(-0.5,0)     
    node[midway,above]{$x$};
\draw [-stealth](f1.east)  -- (f2.west)  
    node[midway,above]{$z_1$};   
\draw [-stealth](f2.east)  -- ++(1,0)  
    node[midway,above](a3){}; 
\draw [stealth-](f3.west) -- ++(-0.5,0)  
 node[midway,above](a4){};
 \draw [-stealth](f3.east) -- ++(0.5,0)  
 node[midway,above]{$z$};
\path (f2) -- node[auto=false]{\ldots} ++(3,0);
\draw [decorate,
    decoration = {calligraphic brace,mirror}] (-0.5,-0.8) --  (5,-0.8) node[midway,below]{$q$ blocks};
\end{tikzpicture}
\end{center}
Let $z_1$ be the output of the first filter, then we can write $z_1=x(1-\frac{a}{s+a})$. Further,  we have
\begin{align}
    |Z_1(t)|= |L^{-1}(z_1)|<|X(t)|\hspace{-0.1em}+\hspace{-0.1em}a\hspace{-0.2em}\int_0^t\hspace{-0.2em} e^{-a(t-\tau)}X(\tau)d\tau.\label{sz2}
\end{align}
For the second term of \eqref{sz2}, performing a similar analysis as done in Lemma \ref{lem1}, we have $ |Z_1(t)|<{X_0}\left(\frac{2a-\mu_x}{a-\mu_x}\right)e^{-\mu_xt}+2X_\infty.$
Recursively following the above steps $(q-1)$ times, we have $|Z(t)|<{X_0}\left(\frac{2a-\mu_x}{a-\mu_x}\right)^qe^{-\mu_xt}+2^qX_\infty.$
\end{proof}
\smallskip
\begin{corollary}\label{coro1_2}
If $z=\frac{s^q}{(s+a)^p}x,$ and, $p\ge q$ and $p,q\in \mathbb{Z}^+,$ then $|Z(t)|<Z_0e^{-\mu_xt}+Z_\infty$, with $Z_0={X_0}\frac{(2a-\mu_x)^q}{(a-\mu_x)^p}$ and $Z_\infty=\frac{2^q}{a^{p-q}}X_\infty.$
\end{corollary}
\begin{proof}
Following figure below and using Lemma \ref{lem1}, we can easily obtain the bounds of $Z_1(t)$.
\smallskip
\begin{center}
\begin{tikzpicture}
    
\node[draw,
minimum width=1cm,
minimum height=1cm](f2){$\frac{1}{s+a}$};

\node[draw,
minimum width=1cm,
minimum height=1cm,
right=1cm of f2](f3){$\frac{1}{s+a}$};

\node[draw,
minimum width=1cm,
minimum height=1cm,
right=0.5cm of f3](f4){$\frac{s}{s+a}$};

\node[draw,
minimum width=1cm,
minimum height=1cm,
right=1cm of f4](f5){$\frac{s}{s+a}$};

\draw [stealth-](f2.west)  -- ++(-0.5,0)  
    node[midway,above]{$x$};   
\draw [-stealth](f2.east)  -- ++(0.25,0)  
    node[midway,above](a3){}; 
\draw [stealth-](f3.west) -- ++(-0.25,0)  
 node[midway,above](a4){};
 \draw [-stealth](f3.east) -- ++(0.5,0)  
 node[midway,above]{$z_1$};
\path (f2) -- node[auto=false]{\ldots} ++(1.5,0);
\draw [decorate,
    decoration = {calligraphic brace,mirror}] (-0.5,-0.8) --  (2.5,-0.8) node[midway,below]{$(p-q)$ blocks};
    
\draw [-stealth](f4.east)  -- ++(0.25,0)  
    node[midway,above](a3){}; 
\draw [stealth-](f5.west) -- ++(-0.25,0)  
 node[midway,above](a4){};
 \draw [-stealth](f5.east) -- ++(0.5,0)  
 node[midway,above]{$z$};
\path (f4) -- node[auto=false]{\ldots} ++(1.5,0);
\draw [decorate,
    decoration = {calligraphic brace,mirror}] (3,-0.8) --  (6,-0.8) node[midway,below]{$q$ blocks};

\end{tikzpicture}
\end{center}
Further, using Corollary \ref{cor1}, it is  straightforward to prove the given result.
\end{proof}

    



\smallskip
\begin{lemma}\label{blemma}
     If $|r|<\psi_{r}$ and  $\lambda_i=\binom{n-1}{n-i}a^{n-i}$, $a>\mu_r$ is a positive design constant, then   for all $ t\in\mathbb{R}_0^+$ and $\forall i\in\{0,1,\ldots, n-1\}$,
     \begin{align}\label{rilde}
         |\tilde\xi^{(i)}(t)|<\frac{(2a-\mu_r)^{i}\psi_{r0}}{(a-\mu_r)^{n-1}}e^{-\mu_rt}+\frac{2^i\psi_{r\infty}}{a^{n-i-1}}.
     \end{align}  
     .
\end{lemma}

\begin{proof}
Using the Laplace transformation, \eqref{r} can be written as $ L\{\tilde \xi(t)\}=\frac{L\{r\}}{(s+a)^{n-1}}+\sum_{k=1}^{n-1}(\frac{1}{s^k}-\frac{\sum_{i=1}^{k}\binom{n-1}{n-i}a^{n-i}s^{(i-1)}}{s^k(s+a)^{n-1}})\tilde\xi^{(k-1)}(0).$ Further, we can have $L\{\tilde \xi^{(i)}(t)\}=s^iL\{\tilde \xi(t)\}-\sum_{k=0}^{i-1}s^{i-k-1}\tilde\xi^{(k)}(0),~\forall i \in \mathbb{N}_{n-1}.$
For the sake of simplicity, the proof will be restricted to the situation of $\tilde \xi^{(i)}(0)=0,    ~\forall i\in\{0,1,\ldots, n-1\}$. Now, the general expression for $L\{\tilde \xi^{(i)}(t)\}, i\in\{0,1,\ldots, n-1\},$ can be written as 
\begin{align}
    L\{\tilde \xi^{(i)}(t)\}=\frac{s^iL\{r\}}{(s+a)^{n-1}}. \label{winitial}
\end{align}
Following the hypothesis and Corollary \ref{coro1_2}, one can deduce from \eqref{winitial} that $\forall t\in\mathbb{R}_0^+,$  and $i\in\{0,1,\ldots, n-1\}$,
    $|\tilde\xi^{(i)}(t)|<\frac{(2a-\mu_r)^{i}\psi_{r0}}{(a-\mu_r)^{n-1}}e^{-\mu_rt}+\frac{2^i\psi_{r\infty}}{a^{n-i-1}}.$
\end{proof}
\smallskip
\begin{remark}
    In \eqref{rilde},  substituting the parameter given in \eqref{mu} -  \eqref{psirinf} for $i=0$, yields $|\tilde\xi|<\psi_0 e^{-\mu t}+\psi_\infty$, or $|\tilde\xi|<\psi$,  one of our control goals. Hence, if we can make filtered tracking error $r$ to follow its VPC $\psi_r$, or the hypothesis of the above lemma, i.e., $|r|<\psi_r$, then the goal will be achieved. To achieve the same, control input is designed in \eqref{u}, based on filtered tracking error, VPC and PIC. Next,  a few results based on the above lemma is presented, and further, a few lemmas will be established to aid the stability analysis in a subsequent section. 
\end{remark}
\smallskip
\begin{corollary}\label{colo2_1}
If $|r|<\psi_{r}$ and  $\lambda_i=\binom{n-1}{n-i}a^{n-i}$, $a>\mu_r$ is a positive design constant, then  in \eqref{r2},
 \begin{align}
     |\phi|&<\frac{\psi_{r0}\bar c_1}{(a-\mu_r)^{n-1}}+\frac{\psi_{r\infty}\bar c_2}{(a)^{n-1}},\label{phibound}
 \end{align}
 where $\bar c_1=(2a-\mu_r)\left((3a-\mu_r)^{n-1}-(2a-\mu_r)^{n-1}\right)$ and $\bar c_2=2a\left((3a)^{n-1}-(2a)^{n-1}\right)$.
\end{corollary}
\begin{proof}
Following $\phi$ in \eqref{r2} and using Lemma \ref{blemma}, we have $|\phi|<\frac{\psi_{r0}}{(a-\mu_r)^{n-1}}\sum_{i=1}^{n-1}\binom{n-1}{n-i}a^{n-i}(2a-\mu_r)^i + \frac{\psi_{r\infty}}{a^{n-1}}\sum_{i=1}^{n-1}\binom{n-1}{n-i}a^{n-i}(2a)^i$. Further following the identity given in Appendix \ref{appendix1}, i.e.,  $\sum_{i=1}^{n-1}\binom{n-1}{n-i}a^{n-i}h^{i}=h((a+h)^{n-1}-h^{n-1})$, with $h=2a-\mu_r$ and $h=2a$, for the first and second terms, respectively, one can readily obtain \eqref{phibound}.
\end{proof}
\smallskip
\begin{corollary}\label{coro2_2}
If $|r|<\psi_{r}$ and  $\lambda_i=\binom{n-1}{n-i}a^{n-i}$, $a>\mu_r$ is a positive design constant, then in \eqref{r2},  
\begin{align}
       |f\left({\bm \xi}(t)\right)|\hspace{-0.1em}<\hspace{-0.1em} k_ln^{1/p^*}\hspace{-0.2em}\left(\hspace{-0.2em}\frac{(2a-\mu_r)^{n-1}\psi_{r0}}{(a-\mu_r)^{n-1}}\hspace{-0.1em}+\hspace{-0.1em}\frac{(2a)^{n-1}\psi_{r\infty}}{a^{n-1}}\hspace{-0.1em}+\hspace{-0.1em}\bar \xi_\mathsf{d}\hspace{-0.2em}\right).\label{fbound}
 \end{align}
\end{corollary}

\begin{proof}
Since $\bm {\xi}(t)\in \mathbb{R}^n$, i.e., finite-dimensional vector space, so all norms are equivalent or one can find constant $c_1$ such that $||\bm {\xi}(t)||_{p^*}\le c_1||\bm {\xi}(t)||_\infty,$ for all $t\in\mathbb{R}_0^+$ Further, using Holder inequality, one can find $c_1=n^{1/p^*}$, holds the equivalence relation. Now using the Assumption \ref{af}, we have 
\begin{align}\label{fineq}
  |f\left({\bm \xi}(t)\right)|\le k_ln^{1/p^*}||\bm {\xi}(t)||_\infty.  
\end{align} Let $\bm{\tilde \xi}= [\tilde \xi, \dot{\tilde\xi}, \ldots, {\tilde\xi}^{(n-1)}]^T$ and $\bm{\xi_\mathsf{d}}=[\xi_\mathsf{d}, \xi_\mathsf{d}^{(1)}, \hdots,\xi_\mathsf{d}^{(n-1)}]^T,$ then following \eqref{tilde}, we have $\bm{\tilde \xi} = \bm{ \xi}-\bm{\xi_\mathsf{d}}$. Substituting $\bm{ \xi}=\bm{\tilde \xi}+\bm{\xi_\mathsf{d}}$ in \eqref{fineq}, and applying triangular inequality, we have $|f\left({\bm \xi}(t)\right)|\le k_ln^{1/p^*}(||\bm {\tilde \xi}(t)||_\infty+||\bm {\xi_\mathsf{d}}(t)||_\infty)$, for all $t\in\mathbb{R}_0^+$. Further using the Lemma \ref{blemma} and Assumption \ref{ade}, one can readily obtain \eqref{fbound}.
\end{proof}
\smallskip
\begin{corollary} \label{cor2.3}
If $|r|<\psi_{r}$ and  $\lambda_i=\binom{n-1}{n-i}a^{n-i}$, $a>\mu_r$ is a positive design constant, then  
\begin{equation}\label{rdot}
    \begin{split}
         \dot r&<\psi_{r0}c_1+\psi_{r\infty}c_2+\bar\xi_\mathsf{d}c_3+\bar d+ g\upsilon, \text{and}\\
    \dot r&>-\psi_{r0}c_1-\psi_{r\infty}c_2-\bar\xi_\mathsf{d}c_3-\bar d+g\upsilon. 
    \end{split}
\end{equation}
 where  $c_1=\frac{\bar c_1+k_ln^{1/p^*}(2a-\mu_r)^{n-1}}{(a-\mu_r)^{n-1}}$, $c_2=\frac{\bar c_2+k_ln^{1/p^*}(2a)^{n-1}}{a^{n-1}}$, and $c_3=k_ln^{1/p^*}+1$ are positive constants.
\end{corollary}
\begin{proof}
It is straightforward to write
\begin{align}
    |\phi+  f\left({\bm \xi}\right)+ d-\xi_\mathsf{d}^{(n)}|\le  |\phi|+  |f\left({\bm \xi}\right)|+ |d|+|\xi_\mathsf{d}^{(n)}|.
\end{align}
Using the corollaries \ref{colo2_1} and \ref{coro2_2}, and following assumptions  \ref{ad} and \ref{ade}, one can have the following inequality
\begin{align}
    |\phi+  f\left({\bm \xi}\right)+ d-\xi_\mathsf{d}^{(n)}|< \psi_{r0}c_1+\psi_{r\infty}c_2+\bar\xi_\mathsf{d}c_3+\bar d. \label{phifd}
\end{align}
 Further, using \eqref{phifd} in \eqref{r2}, one gets \eqref{rdot}.
\end{proof}
\smallskip
\begin{lemma}\label{lemmainf}
If the filtered tracking error $r$ given in \eqref{r} is transgressing its upper bound mentioned in \eqref{constr}, then  $(r-\psi_{r})$ will approach $0$ from the left side and 
 \begin{align}\label{elemmainf}
   \lim_{(r-\psi_r)\uparrow 0}{\dot r}\ge-\mu_r\psi_{r0}.  
 \end{align}
\end{lemma}
\begin{proof}
It is straightforward to assume that before transgressing any bounds, the tracking error must be within its prescribed performance bounds (i.e., $-\psi_r<r<\psi_r$). This implies that 
$    -2\psi_r<r-\psi_r<0$.
Thus, we can analyze that if $r$ is transgressing its upper bound, i.e., $\psi_r$ then $(r-\psi_r)$ will approach $0$ from the left side. Consequently, it is easy to know that when $(r-\psi_r)$ approaches $0$ from the left side, the time derivative of $(r-\psi_r)$ will be greater than equal to $0$. As a result, we have
\begin{align}\label{zphi}
    \lim_{(r-\psi_r)\uparrow 0}{\dot r}\ge\dot\psi_r.
\end{align}
Noting \eqref{psidb}, we can infer from \eqref{zphi} that
 $\lim_{(r-\psi_r)\uparrow 0}{\dot r}\ge-\mu_r\psi_{r0}$.
\end{proof}
\begin{lemma}\label{lemmasup}
If the tracking error is transgressing its lower bound, then $(r+\psi_r)$ will approach $0$ from the right side, and 
 \begin{align}\label{elemmasup}
   \lim_{(r+\psi_r)\downarrow 0}{\dot r}\le\mu_r\psi_{r0}.  
 \end{align}
\end{lemma}
\begin{proof}
The proof is similar to that of Lemma \ref{lemmainf}.
\end{proof}

\section{Stability analysis}
In this section, stability analysis will be shown based on the results of the lemmas presented in the previous section.

\smallskip
\begin{theorem}\label{theorem1}
Consider a system \eqref{sys1}, with desired state trajectory $\xi_\mathsf{d}$, PPC on output tracking error, $\psi=\psi_0e^{-\mu t}+\psi_\infty$ and a PIC $\bar \upsilon$. If the system \eqref{sys1} satisfies  Assumption \ref{af}-\ref{ade}, and  the control input is designed as in \eqref{u}, then  system output will follow its desired trajectory, tracking error and input will never transgress its PPC and PIC, respectively, and all the closed-loop signals will remain bounded, provided the following feasibility conditions for PIC and PPC are true.
\begin{align}
   &\text{PIC:}~ \bar\upsilon>\frac{1}{\barbelow g}(\psi_{r\infty}c_2+\bar\xi_\mathsf{d}c_3+\bar d),\label{cond}\\
   &\text{PPC:}~ |r(0)|(a\hspace{-0.1em}-\hspace{-0.1em}\mu)^{1-n}\hspace{-0.1em}<\hspace{-0.1em}\psi_0\hspace{-0.1em}<\hspace{-0.1em}\frac{\barbelow g \bar\upsilon \hspace{-0.1em}-\hspace{-0.1em}\psi_{r\infty}c_2\hspace{-0.1em}-\hspace{-0.1em}\bar\xi_\mathsf{d}c_3\hspace{-0.1em}-\hspace{-0.1em}\bar d}{(c_1\hspace{-0.1em}+\hspace{-0.1em}\mu)(a\hspace{-0.1em}-\hspace{-0.1em}\mu)^{n-1}}.\label{cond1}
\end{align}

\end{theorem}

\begin{proof}
Stability analysis is done using proof-by-contradiction. To begin with the proof, we will first establish proposition \textit{P}$1$ as follows.


\textit{P}$1$: If the \eqref{cond} and \eqref{cond1} is true, input is designed as \eqref{u}, and initially $r$ is within its designed constraints $\psi_r$, then there exists at least a time instant at which $r$ violates its constraints, or,
 $$ \exists~ t_j \text{~such that~} |r(t_j)|>\psi_r(t_j), \forall t_j \in  (t_1,\ldots, t_i, \ldots, t_{\bar n} ),$$ 
 where $t_i<t_{i+1}$, $t_i$ represent $i${th} instant of violation of performance constraint, $i \in \mathbb{N}$, and $\bar n \in \mathbb{N}$. 
We are now prepared for the proof. Suppose that  \textit{P}$1$ is true, then we have the following.
\begin{align}\label{zt1}
    |r(t)|<\psi_r(t),~ \forall t\in [0,t_1).
\end{align}
Suppose that at the instant of time $t_1$, the tracking error is transgressing its performance constraints (i.e., upper or lower bounds). With the following analysis, we will see that error never transgresses its performance constraints. 

Noting \eqref{zt1}, and using \eqref{rdot} of Corollary \ref{cor2.3}, for all $t\in [0,t_1)$, we have
\begin{align}
         \dot r&<\psi_{r0}c_1+\bar\xi_\mathsf{d}c_2+g\upsilon+\bar d, \text{and} \label{proof1}\\
    \dot r&>-\psi_{r0}c_1-\bar\xi_\mathsf{d}c_2+g\upsilon-\bar d. \label{proof2}
\end{align}

Following \eqref{u}, we infer that
\begin{align} 
    \liminf_{(z-\psi)\uparrow0}{\upsilon}&=-\bar \upsilon,\label{liminff}\\
    \limsup_{(z+\psi)\downarrow0}{\upsilon}&=\bar \upsilon\label{limsupp}.
\end{align}
Consequently, following Assumption \ref{ag}, we have 
\begin{align} 
    -\bar g\bar \upsilon&\le\liminf_{(r-\psi_r)\uparrow0}{g\upsilon}\le-\barbelow g\bar \upsilon,\label{liminf}\\
   \barbelow g\bar \upsilon&\le\limsup_{(r+\psi_r)\downarrow0}{g\upsilon}\le\bar g\bar \upsilon.\label{limsup}
\end{align}
Using \eqref{liminf} and \eqref{proof1}, we can infer that for all $ t\in [0,t_1),$
\begin{align}\label{infproof}
\liminf_{(r-\psi_r)\uparrow 0}{\dot r}<\psi_{r0}c_1+\psi_{r\infty}c_2+\bar\xi_\mathsf{d}c_3-\barbelow g\bar\upsilon+\bar d.
\end{align}
Similarly, using \eqref{limsup} and \eqref{proof2}, $\forall t\in [0,t_1),$ we obtain  
\begin{align}\label{supproof}
\limsup_{(r+\psi_r)\downarrow 0}{\dot r}>-\psi_{r0}c_1-\psi_{r\infty}c_2-\bar\xi_\mathsf{d}c_3+\barbelow g \bar\upsilon-\bar d.
\end{align}
Now, recalling \eqref{mu}, \eqref{psir0}, and \eqref{cond1}, it follows
\begin{align}
    \psi_{r0}(c_1+\mu_r)<\barbelow g \bar\upsilon-\psi_{r\infty}c_2 -\bar\xi_\mathsf{d}c_3-\bar d. \label{step}
\end{align}
Further, \eqref{step} can be written as,\vspace{-0.2cm}
\begin{align}
  \psi_{r0}c_1+\psi_{r\infty}c_2+\bar\xi_\mathsf{d}c_3-\barbelow g\bar\upsilon+\bar d&<-\mu_r\psi_{r0}, \text{or}  \label{step1}\\
  -\psi_{r0}c_1-\psi_{r\infty}c_2-\bar\xi_\mathsf{d}c_3+\barbelow g \bar\upsilon-\bar d&>\mu_r\psi_{r0}.\label{step2}
\end{align}
 Now incorporating \eqref{step1} in \eqref{infproof}, and \eqref{step2} in \eqref{supproof}, it can inferred that over $[0, t_1)$ \vspace{-0.3cm}
\begin{align}
\liminf_{(r-\psi_r)\uparrow 0}{\dot r}&<-\mu_r\psi_{r0},\label{infproof1}\\
\limsup_{(r+\psi_r)\downarrow 0}{\dot r}&>\mu_r\psi_{r0}. \label{supproof1}
\end{align}
Recalling lemmas \ref{lemmainf} and \ref{lemmasup}, it can be  inferred that \eqref{infproof1} contradicts \eqref{elemmainf}, and \eqref{supproof1} contradicts \eqref{elemmasup}. Hence, over $[0, t_1)$, tracking error will never approach its  performance constraints. Consequently, it can be concluded that there is no $t_1$ in which the $r$ violates its designed constraint $\psi_r$. Since there does not exist the first instant of violation of the designed constraint, there does not exist any time at which $r$ will violate its constraints $\psi_r$. Therefore, it can be concluded that  \textit{P}$1$ is false.
Now following \eqref{mu}, \eqref{psir0} and \eqref{cond1}, it follows $\psi_{r0}>|r(0)|$, and following \eqref{constr}, $\psi_{r}(0)>|r(0)|$. Thus
 initially $r$ is within its designed VPC ($\psi_r$), and further noting that Proposition \textit{P}$1$ is false, we have the following.
\begin{align}\label{ztf}
    |r(t)|<\psi_r(t),~ \forall t\ge0.
\end{align}
Now, following Lemma \ref{blemma} and using \eqref{mu}-\eqref{lambda}, we have
\begin{align}\label{finalbound}
    |\tilde\xi^{(i)}(t)|\hspace{-0.2em}<\hspace{-0.2em}(2a-\mu_r)^{i}\psi_{0}e^{-\mu_rt}\hspace{-0.2em}+\hspace{-0.2em}{(2a)}^i\psi_{\infty},~ i\hspace{-0.1em}\in\hspace{-0.1em}\{0,1,\ldots, n\hspace{-0.2em}-\hspace{-0.2em}1\}.
\end{align}
Using \eqref{finalbound}, it can be concluded that $\tilde\xi ^{(i)}$ will converge asymptotically to a set, $\Gamma_{i}:=\{\tilde\xi ^{(i)}\in \mathbb{R}:|\tilde\xi ^{(i)}|<{(2a)}^i\psi_{\infty}\}$ and output tracking error $\tilde\xi$ will follow its PPC, i.e., $\psi(t)$, $\forall t\in \mathbb{R}^{+}_0$.
Now, we will seek the boundedness of all the closed-loop signals.

Following \eqref{ztf} and \eqref{finalbound} we have $r \in \mathcal{L}^\infty $ and $\tilde \xi ^{(i)}\in \mathcal{L}^\infty$. Consequently, following assumption \ref{ade} and recalling from \eqref{tilde} that $\xi_i=\tilde\xi_i^{(i-1)}+\xi_\mathsf{d}^{(i-1)}$,  we have  $\xi_i \in \mathcal{L}^\infty,  ~\forall i \in \mathbb{N}_{n}$. Knowing the fact that $f(\bm\xi)$ in \eqref{sys1} is smooth nonlinear function, as a  result, we have $f(\bm\xi) \in  \mathcal{L}^\infty$. Also, it is straightforward to follow from \eqref{u} that if $|r|<\psi_r$, then $|\upsilon|<\bar \upsilon$, thus we have $\upsilon \in\mathcal{L}^\infty $. Following \eqref{sys1} and \eqref{r2}, and with the  help of established boundedness of the signal, and assumptions \ref{ag} and \ref{ad}, that $g(\bm\xi)$ and disturbance are bounded,  we have $\dot\xi_i \in\mathcal{L}^\infty,  ~\forall i \in \mathbb{N}_{n} $ and $\dot r \in\mathcal{L}^\infty,$ respectively. Thus, all closed-loop signals are bounded. This completes the proof.
\end{proof}
\section{Simulation Results and Discussion}
In this section, a simulation study is presented to show the effectiveness of the proposed approach. Consider a control-affine nonlinear system
\begin{equation}\label{sys2}
\begin{split}
    \dot \xi_1 &=\xi_2,\\
    \dot \xi_2&=-0.5(\sin{\xi_1}+\xi_2)+(3+\cos{\xi_2})\upsilon+d,\\
    y&=\xi_1,
    \end{split}
\end{equation} 
where $\xi(t)\in \mathbb{R}$, $\upsilon(t)\in \mathsf{U}\in \mathbb{R}$ and $y$ are the state, the input, and the output of the system \eqref{sys2}, respectively, and $d(t)=0.5\sin{2t}$ is a disturbance. The desired output is $\xi_\mathsf{d}(t)=0.5\sin{t}$. For \eqref{sys2}, correspondingly, following \eqref{sys1}, we can note that $f(\xi)=-0.5(\sin{\xi_1}+\xi_2)$ and $g(\xi)=3+\cos{\xi_2}$, and are assumed to be unknown. For \eqref{sys2}, one can readily obtain $k_l=0.5$, $\barbelow g=2$, $\bar d =0.5$, and for the given desired output, we have $\bar \xi_\mathsf{d}=0.5$. The design parameter $a$ is chosen as  $a=2$, accordingly following corollary \ref{cor2.3}, $c_1=9$, $c_2=6$, and  $c_3=2$. Now, following the feasibility conditions \eqref{cond}, we have PIC: $\bar\upsilon> 0.78$. The goal is to design control law $\upsilon$ such that the output tracks the desired trajectory without transgressing PIC:  $\bar \upsilon=6$, and PPC: $\psi(t)=\psi_0e^{-\mu t}+\psi_{\infty}$, (with $\psi_0=1, \psi_{\infty}=0.01$ and $\mu=1$),   on tracking error. It can be easily verified using \eqref{cond1} that PPC satisfies its feasibility condition for a given PIC, i.e. $\psi_0<1.1$. The controller is designed using \eqref{u}, $\upsilon=-\frac{2\bar \upsilon}{\pi}\arctan\left(\frac{\pi}{2\bar\upsilon}\tan\left(\frac{\pi r}{2\psi_r}\right)\right)$, where, as mentioned in  \eqref{r} $r=\lambda_1\tilde\xi_1+\dot{\tilde{\xi}}_1$, with $\tilde{\xi}_1=\xi_1-\xi_\mathsf{d}$ and $\dot{\tilde\xi}_1=\xi_2-\dot\xi_\mathsf{d}$ as mentioned in  \eqref{r}, and $\psi_r=\psi_{r0}e^{-\mu_rt}+\psi_{r\infty}$. The parameter $\mu_r, \psi_{r0}, \psi_{r\infty},$ and $\lambda_1$ are given by \eqref{mu}-\eqref{lambda}, the aforementioned parameters, and the simulation study is done for two sets of initial conditions, i.e.,  $\bm\xi(0)=[0.4 ~0.29]^T$ and $\bm\xi(0)=[0.6 ~0.29]^T.$ For both sets of initial conditions, it can be observed from Fig. \ref{fig1} that the output tracks the desired trajectory along with its tracking error following the PPC. Also, from Fig. \ref{fig2}, it can be seen that that input follows its PIC. Further, it can be observed from Fig. \ref{fig2} the filtered tracking errors follow its VPC. It is to note that, since $\xi_\mathsf{d}(0)=0$ and $\dot\xi_\mathsf{d}(0)=0.5$, so with  change in the initial condition $\bm\xi(0)$ from $[0.4 ~0.29]^T$ to $[0.4 ~0.29]^T$, $\bm\tilde\xi(0)=[\tilde\xi_1(0)~ \dot{\tilde{\xi}}_1(0)]$ changes from $[0.4 ~-0.21]^T$ to $[0.6 ~-0.21]^T$, respectively. Consequently, $r(0)$ changes from $0.59$ to $0.99$, and also $|r(0)|(a-\mu)^{1-n}$ changes from $0.59$ to $0.99$. It can be calculated that a further increase in $\xi_1(0)$ from $0.6$ will violate the feasibility condition $|r(0)|(a-\mu)^{1-n}<\psi_0$, also it can be observed from Fig. \ref{fig2} that initially control input is near to its PIC, thus motivating the feasibility condition. 
The observation made from the Fig. \ref{fig1} and \ref{fig2} was as expected and stated in Theorem  \ref{theorem1}. 
\begin{figure}
    \centering
    \includegraphics[width=8.5cm,height=5cm]{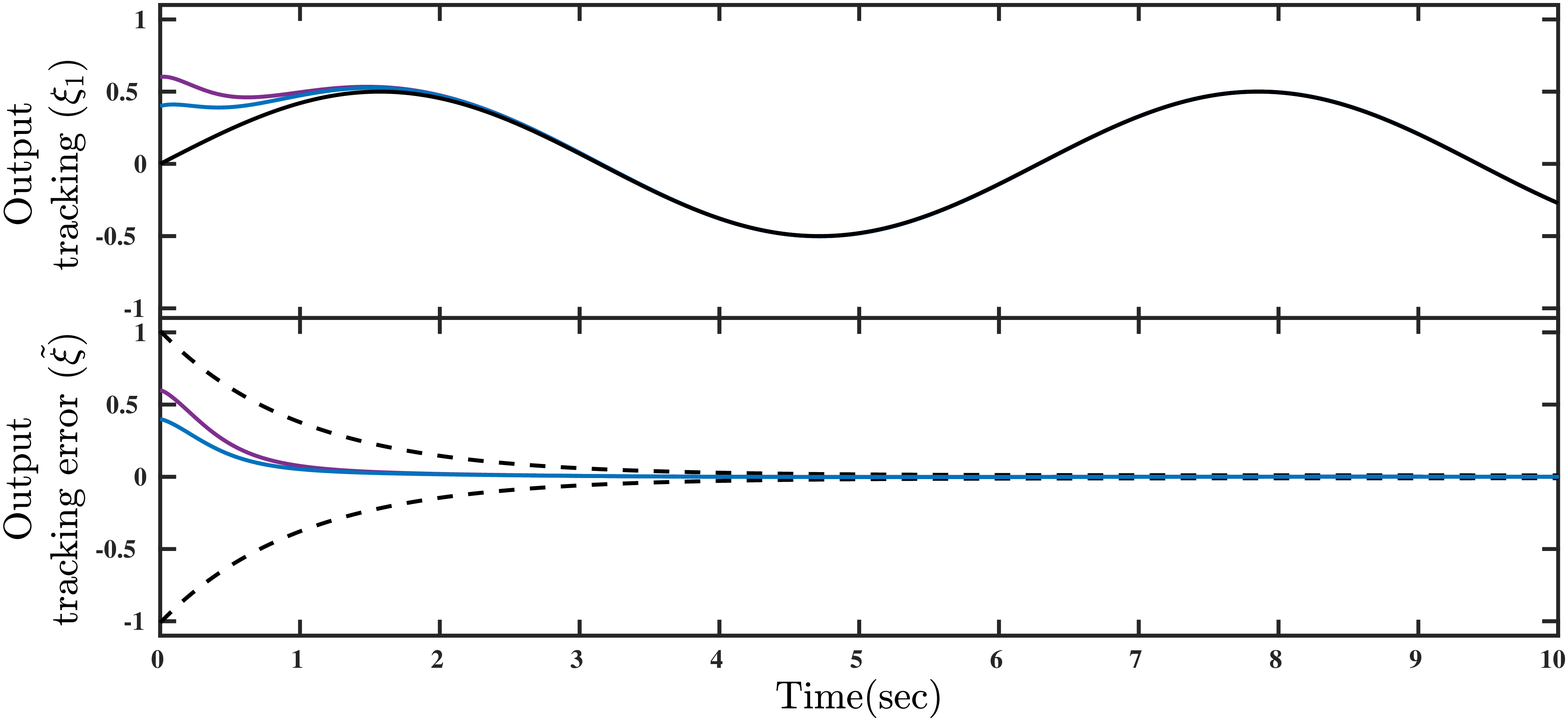}
    \vspace{-0.7cm}
       \caption{Top: output tracking performance (\lep $\xi_1$ for $\bm{\xi}(0)=[0.4 ~0.29]^T$, \leb $\xi_1$ for $\bm{\xi}(0)=[0.6 ~0.29]^T$, \leblack $\xi_\mathsf{d}$(desired output)); Bottom: prescribed performance of tracking error( \ledas PPC $(\psi)$, \lep $\tilde\xi_1$ for $\bm{\xi}(0)=[0.4 ~0.29]^T$, \leb $\tilde\xi_1$ for $\bm{\xi}(0)=[0.6 ~0.29]^T$.) }
    \label{fig1}
 \end{figure}

 \begin{figure}
    \centering
    \includegraphics[width=8.5cm,height=5cm]{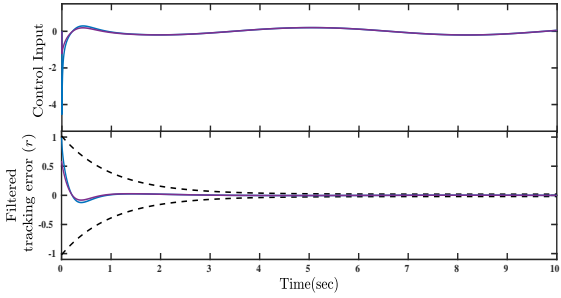}
    \vspace{-0.7cm}
       \caption{Top: control input (\lep $\bar \upsilon$ for $\bm{\xi}(0)=[0.4 ~0.29]^T$, \leb $\bar \upsilon$ for $\bm{\xi}(0)=[0.6 ~0.29]^T$); Bottom: performance of filtered tracking error (\lep $r$ for $\bm{\xi}(0)=[0.4 ~0.29]^T$, {\leb} $r$ for $\bm{\xi}(0)=[0.6 ~0.29]^T$) for the designed VPC ({\ledas $\psi_r$}). }
    \label{fig2}
 \end{figure}

\section{Conclusion}
A controller has been proposed for the tracking problem of control affine nonlinear system subjected to PPC and PIC. The structure of the controller is simple as it does not require any adaptive laws, calculation of any derivatives, system knowledge or approximation. Hence, the controller is easy to implement and an approximation-free controller. Also, the derived feasibility condition for the prescription of constraint restricts arbitrary prescription. The simulation results confirm these facts. In future, the work will be extended for multiagent systems.

 \begin{appendices}
     \section{Proof for $\sum_{i=1}^{n-1}\binom{n-1}{n-i}a^{n-i}h^{i}=h((a+h)^{n-1}-h^{n-1}).$} \label{appendix1}
     Using the binomial  identity $\binom{m}{k}=\binom{m}{m-k}$, we have  $\sum_{i=1}^{n-1}\binom{n-1}{n-i}a^{n-i}h^{i}=\sum_{i=1}^{n-1}\binom{n-1}{i}a^{n-i}h^{i}.$ Further using  the binomial identity $\binom{m}{k}=\binom{m-1}{k}+\binom{m-1}{k-1},$ it can be written as $\sum_{i=1}^{n-1}\binom{n-1}{n-i}a^{n-i}h^{i}=\sum_{i=1}^{n-1}\left(\binom{n}{i}-\binom{n-1}{i-1}\right)a^{n-i}h^{i}.$
Substituting $\sum_{i=1}^{n-1}\binom{n}{i}a^{n-i}h^{i}=(a+h)^{(n)}-h^n-a^n$ and $\sum_{i=1}^{n-1}\binom{n-1}{i-1}a^{n-i}h^{i}=a(a+h)^{(n-1)}-a^n$, it can be written as 
$\sum_{i=1}^{n-1}\binom{n-1}{n-i}a^{n-i}h^{i}=(a+h)^{(n)}-h^n-a(a+h)^{(n-1)}.$ Further simplifying, we have $\sum_{i=1}^{n-1}\binom{n-1}{n-i}a^{n-i}h^{i}=h((a+h)^{n-1}-h^{n-1}).$
 \end{appendices}



\bibliographystyle{ieeetr}
\bibliography{ref1.bib}
\end{document}